\documentclass[12pt,reqno]{amsart}

\usepackage{amsmath,amsthm,amssymb,amsfonts,mathtools}
\usepackage{bm}
\usepackage{natbib}
\usepackage{bbm}
\usepackage{graphicx}
\usepackage{multirow}
\usepackage{here}
\usepackage{fullpage}
\usepackage{xurl}
\usepackage{color}
\usepackage{xcolor}
\usepackage{mathrsfs}
\usepackage{enumitem}
\usepackage{comment}
\usepackage{subcaption}
\usepackage{pifont}
\usepackage[linesnumbered,ruled,vlined]{algorithm2e}

\makeatletter

\newtheorem*{lemma}{Lemma}
\newtheorem{proposition}{Proposition}

\newtheorem*{asm-SPT}{Assumption Scalar-PT}
\newtheorem*{asm-IPT}{Assumption Interval-PT}
\newtheorem*{asm-PS}{Assumption PS}
\newtheorem*{asm-CPS}{Assumption Conditional-PS}
\theoremstyle{definition}

\newtheorem*{cond}{Condition}

\def\aE#1{\mathbb{E}_{\mathrm{A}}\left[#1\right]}

\def\E#1{\mathbb{E}\left[#1\right]}
\def\P#1{\mathbb{P}\left[#1\right]}

\renewcommand{\hat}{\widehat}

\def\bco{\iffalse} 
\allowdisplaybreaks

\makeatother

\begin{document}

\title[DID with Interval Data]{Difference-in-Differences with Interval Data}
\thanks{} 

\author[D. Kurisu]{Daisuke Kurisu}
\author[Y. Okamoto]{Yuta Okamoto}
\author[T. Otsu]{Taisuke Otsu}

\date{First version: \today}

\address[D. Kurisu]{Center for Spatial Information Science, The University of Tokyo, 5-1-5, Kashiwanoha, Kashiwa-shi, Chiba 277-8568, Japan.}
\email{daisukekurisu@csis.u-tokyo.ac.jp}

\address[Y. Okamoto]{Graduate School of Economics, Kyoto University, Yoshida Honmachi, Sakyo, Kyoto 606-8501, Japan.
}
\email{okamoto.yuuta.57w@st.kyoto-u.ac.jp}

\address[T. Otsu]{Department of Economics, London School of Economics, Houghton Street, London, WC2A 2AE, UK.
}
\email{t.otsu@lse.ac.uk}

\begin{abstract}
Difference-in-differences (DID) is one of the most popular tools used to evaluate causal effects of policy interventions. This paper extends the DID methodology to accommodate interval outcomes, which are often encountered in empirical studies using survey or administrative data. We point out that a naive application or extension of the conventional parallel trends assumption may yield uninformative or counterintuitive results, and present a suitable identification strategy, called parallel shifts, which exhibits desirable properties. Practical attractiveness of the proposed method is illustrated by revisiting an influential minimum wage study by Card and Krueger (1994).
\end{abstract}

\maketitle

\section{Introduction}\label{sec:intro}
Difference-in-differences (DID) is widely applied in empirical research to study the causal effect of a policy intervention or unexpected event by exploiting observational panel data. Under the parallel trends assumption (i.e., average outcomes of control and treated groups in the absence of an intervention obey the same trend), DID identifies the average treatment effect on the treated (ATT) by a contrast of average outcomes of those groups for pre- and post-intervention periods. Since the seminal works by \cite{Ashenfelter_Card:1985} and \cite{Card_Krueger:1994}, the DID methodology has been extended to various directions, such as heterogeneous treatment effects \citep{Athey_Imbens:2006}, staggered treatment adoption \citep{CS21}, robustness of the parallel trends assumption for functional forms \citep{Roth_SantAnna:2023}, among others. We refer to \cite{Roth_etal:2023} for an overview of DID methods.

Although these extensions primarily focus on the scalar-outcome case, many outcomes in survey or administrative data are often reported in interval form, such as household incomes, wages, hours worked, and time spent in specific activities. 
Representative examples include income measures in the U.S. Health and Retirement Study, the United Kingdom's Living Costs and Food Survey, and the Australian Census, as well as the hours-worked measure in the UK Census.
Even when the dataset is seemingly scalar-valued, respondents may report rounded values in surveys (e.g., \citealp{Eissa_Liebman:1996}), in which case it may be more appropriate to treat the outcome as interval data. Although empirical methods on interval data have been extensively studied in the literature of partial identification (see, e.g., \cite{BeMoMo12}, \cite{Manski:2003}, and \cite{Morinari:2020} for overviews), to the best of our knowledge, the literature lacks a guideline for empirical research on how to conduct causal inference for DID designs with interval outcomes.

In this paper, we extend the scope of the DID methodology to accommodate interval-valued outcomes. 
In particular, we study several (partial) identification strategies for the ATT under DID designs when researchers observe only interval-valued outcomes. We first directly apply the conventional parallel trends assumption to the unobservable scalar outcome, and illustrate that the implied sharp identified set for the ATT can be uninformative or counterintuitive via an influential example by \cite{Card_Krueger:1994}. We next investigate an extended notion of the parallel trends assumption, in which the same mapping that transports the bounds of the untreated potential mean from the pre-treatment to the post-treatment period is applied to both the control and treatment groups. Although this seems a natural extension of the conventional parallel trends, this identification strategy also does not work well. For the \cite{Card_Krueger:1994} example, the (estimated) sharp identified set yields an increasing trend in the lower bound for the treated group, even though the lower bound for the control group exhibits a decreasing trend. The first message of this paper is that researchers need to be cautious of a naive application or extension of the parallel trends assumption to interval outcomes, which may yield uninformative or counterintuitive results.

As the second message, we propose a new identification strategy, called parallel shifts, by exploiting an alternative interpretation of the conventional parallel trends assumption, i.e., the shift in the untreated potential mean from the control to the treatment group is identical over time. Our parallel shifts assumption is obtained by applying this interpretation to the bounds of the potential outcomes. We present some desirable properties of the sharp identified set implied by the parallel shifts, which are not shared by alternative identification strategies, and demonstrate its sensible finite sample feature by the \cite{Card_Krueger:1994} example.

This paper contributes to the literature on DID and partial identification. Since this literature is vast, we refer to the surveys cited above and references therein. At the intersection of these literatures, \cite{Manski_Pepper:2018} and \cite{Rambachan_Roth:2023} introduced partial identification analysis for DID when the parallel trends assumption might be violated. Relatedly, \cite{Zhou_etal:2025gdid} extended DID analysis to the case where outcomes are situated in a general metric space. However, these papers do not consider interval-valued outcomes.

This paper is organized as follows. After closing this section with our running example on \cite{Card_Krueger:1994}, Section \ref{sec:setup} discusses two potential identification strategies (scalar and interval parallel trends) and their issues. In Section \ref{sec:shift}, we present the parallel shifts assumption and discuss its desirable properties and implementation. Section \ref{sec:CK} revisits the \cite{Card_Krueger:1994} example and illustrates our main results. All proofs are contained in the appendix.

\subsection*{Running example: \cite{Card_Krueger:1994}}
We use the \cite{Card_Krueger:1994} minimum wage study as a running example throughout the paper. In this early and influential study, the authors examine the effect of minimum wage increases on employment. Specifically, they exploit the minimum wage increase in New Jersey as a natural experiment and compare employment before and after the policy change between New Jersey (NJ) and Pennsylvania (PA) using the DID approach.

Their main outcome variable is the number of full-time-equivalent (FTE) employees in chain fast-food restaurants across the two states. The FTE employment is defined as the number of full-time workers plus the number of store managers plus 0.5 times the number of part-time workers \citep[p.~775]{Card_Krueger:1994}. The underlying employment data were collected by a telephone survey.

Although the original study treats this measure as a scalar outcome, there are several reasons why it should instead be viewed as interval-valued data. First, the dataset contains decimal-valued entries (e.g., $6.5$) even before converting the raw numbers into FTE employees. This suggests that the underlying number of workers was uncertain between 6 and 7, or $[6,7]$.
Furthermore, the reported numbers are likely subject to rounding. 
As graphical evidence, Figure \ref{fig: rounding} shows histograms of full-time and part-time employees in the pre-treatment period. Both histograms exhibit pronounced heaping at multiples of five above ten (i.e., $10, 15, 20,\ldots$). Therefore, these values may be better interpreted as interval-valued. In the subsequent analysis, we regard an observed count $k(\geq10)$ that is a multiple of five as representing the interval $[k-5, k+5]$ (e.g., $15$ as $[10, 20]$). All other integer-valued observations are retained in their original scalar form. In the remainder of this paper, we treat these recoded observations as interval-valued data and reanalyze the effect of the minimum wage on employment.
\begin{figure}[t]
    \centering
    \begin{subfigure}[b]{0.475\textwidth}
        \centering
        \includegraphics[width=\linewidth]{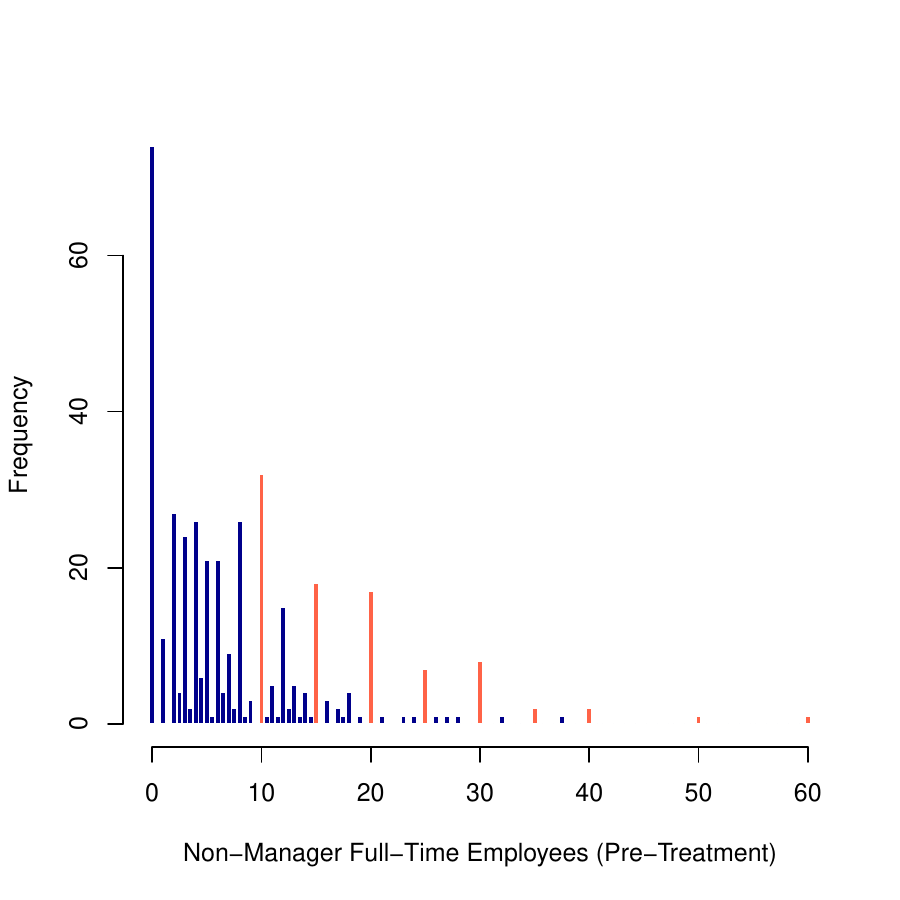}
        \caption{}
        \label{fig: rounding ft}
    \end{subfigure}
    \hfill
    \begin{subfigure}[b]{0.475\textwidth}
        \centering
        \includegraphics[width=\linewidth]{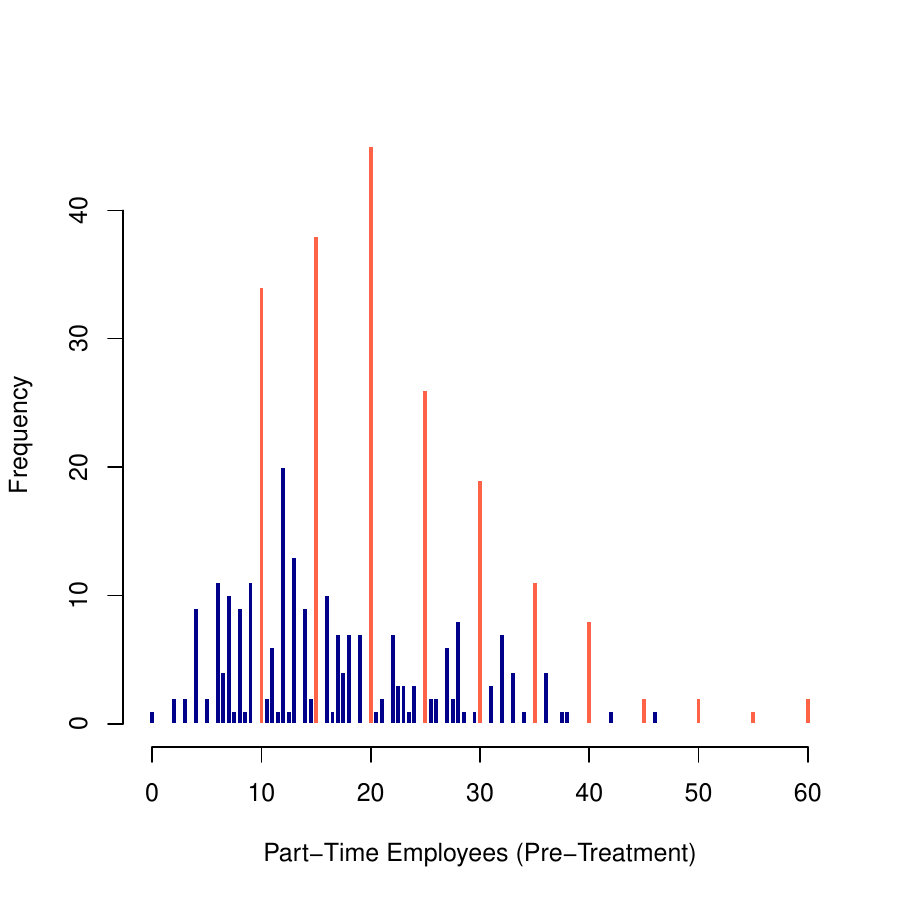}
        \caption{}
        \label{fig: rounding pt}
    \end{subfigure}
    \caption{Evidence of rounding in the \cite{Card_Krueger:1994} data}
    \label{fig: rounding}
\end{figure}

\section{DID with interval data: Potential strategies and their issues}\label{sec:setup}
This paper considers the canonical DID setting with two periods. We observe panel data $(D_{it},\bm{Y}_{it})$ for $i=1,\ldots,n$ and $t=1,2$, where $D_{it}$ is a treatment indicator taking one if $i$ is treated in period $t$, and zero otherwise. Suppose no units are treated in period $t=1$ and some (but not all) units become treated in period $t=2$. We write $D_i=D_{i2}$ hereafter for notational ease. Let $Y_{it}$ denote the scalar outcome variable (i.e., the actual number of employees at restaurant $i$ in period $t$). In contrast to the conventional settings, researchers cannot observe $Y_{it}$. Instead, they observe the interval-valued outcome $\bm{Y}_{it} = [Y_{it}^\ell, Y_{it}^u]$ containing $Y_{it}$ with probability one, and wish to identify causal effects of the treatment $D_i$. 

Let $Y_{it}(d)$ be the usual potential outcome given $D_{i}=d$, and $\bm{Y}_{it}(d)$ be the potential interval containing $Y_{it}(d)$ with probability one. Then we assume
\begin{equation}\label{eq:Y}
\begin{aligned}
    Y_{i1} & = Y_{i1}(0),\qquad
    Y_{i2} = D_{i}Y_{i2}(1) + (1-D_{i}) Y_{i2}(0),  \\
    \bm{Y}_{i1} & = \bm{Y}_{i1}(0),\qquad
    \bm{Y}_{i2} = D_{i}\bm{Y}_{i2}(1) + (1-D_{i}) \bm{Y}_{i2}(0),
\end{aligned}
\end{equation}
wherein the so-called ``no anticipation'' assumption is implicitly imposed. Another implicit assumption here is that the conditional on $D_i$, the way the scalar outcome is encoded into the interval data is invariant for both potential and observable outcomes. We are interested in the average treatment effect on the treated (ATT):
\begin{align*}
    \theta_{\mathtt{ATT}} = \E{Y_{i2}(1)-Y_{i2}(0) \mid D_i=1} = \E{Y_{i2} \mid D_i=1} - \E{Y_{i2}(0) \mid D_i=1},
\end{align*}
where the second equality uses \eqref{eq:Y}. In contrast to the scalar outcome case, even the first term of the right-hand side is not identified.

However, we can still partially identify the first term. Since $Y_{i2}\in\bm{Y}_{i2}$, the sharp identified set for $\E{Y_{i2} \mid D_i=1}$ is given by 
\begin{align}
   \aE{\bm{Y}_{i2} \mid D_i=1} \coloneqq \left[\mathbb{E}[Y_{i2}^\ell \mid D_i=1], \E{Y_{i2}^u \mid D_i=1}\right].\label{eq: bounds on the identified term}
\end{align}
The sharpness follows by considering the cases with $Y_{i2} = \alpha Y_{i2}^\ell + (1-\alpha)Y_{i2}^u$ for $\alpha\in[0,1]$, where the lower and upper bounds are attained when $\alpha=1$ and $0$, respectively.

The interval $\aE{\bm{Y}_{i2} \mid D_i=1}$ is called the conditional Aumann mean of the interval-valued random variable $\bm{Y}_{i2}$.\footnote{Formally, the conditional Aumann mean of an interval-valued random variable $\bm{Y}$ given $D$ is defined as
\begin{align*}
    \aE{\bm{Y} \mid D} = \{\E{Y\mid D} : Y \in L^1 \cap M(\bm{Y})\},
\end{align*}
where $L^1$ is the space of scalar random variables with finite absolute mean and $M(\bm{Y})$ is the set of all measurable points of $\bm{Y}$.} The notion of the Anmann mean is commonly applied in the literature of partial identification (see \cite{BeMo08} and \citet[Chapter 3]{MoMo18} for further details), and we also employ this notion to simplify the presentation. For interval-valued data, the Aumann mean is simply the interval given by the conventional means of the lower and upper bounds.

With this notion of the Aumann mean, we can also obtain the partial identification of the related objects for the control group, i.e., $\E{Y_{it} \mid D_i=0}\in\aE{\bm{Y}_{it}\mid D_i=0}$ for $t=1,2$. The main challenge to partially identify the ATT is to bound the counterfactual mean, $\E{Y_{i2}(0) \mid D_i=1}$. To this end, we begin by considering two seemingly natural but potentially unattractive identification strategies. These considerations motivate our proposal, developed in Section \ref{sec:shift}, as an empirically appealing alternative.

\subsection{Strategy 1: Scalar parallel trends}
The most direct approach to bound the counterfactual mean, $\E{Y_{i2}(0) \mid D_i=1}$, is to rely upon the standard parallel trends assumption (e.g., \citealp[Assumption 1]{Roth_etal:2023}) on the scalar potential outcome $Y_{i2}(0)$, that is
\begin{asm-SPT}\label{assumption: standard parallel trend}
    $\E{Y_{i2}(0) - Y_{i1}(0)\mid D_i=1} = \E{Y_{i2}(0) - Y_{i1}(0)\mid D_i=0}$.
\end{asm-SPT}
This is the conventional parallel trends assumption when the scalar outcome is observable. In this case, the sharp identified set of the counterfactual mean and ATT are obtained as follows. Let $\bm{A}\ominus \bm{B} = \{a-b: a\in \bm{A}\,\text{ and }\, b\in \bm{B}\}$ be the Minkowski difference for intervals $\bm{A}$ and $\bm{B}$.
\begin{proposition}[Sharp identified set under scalar parallel trends]\label{prop: naive bounds}
Consider the setup of this section. Under Assumption Scalar-PT, the sharp identified set of the counterfactual mean $\E{Y_{i2}(0) \mid D_i=1}$ is given by
    \begin{align*}
        \mathcal{M}_{\mathtt{SPT}}& = \bigg[
        \mathbb{E}[Y_{i1}^\ell\mid D_i=1] + \left(\mathbb{E}[Y_{i2}^\ell\mid D_i=0] - \mathbb{E}[Y_{i1}^u\mid D_i=0]\right),\\
        &\qquad\quad 
        \mathbb{E}[Y_{i1}^u\mid D_i=1] + \left(\mathbb{E}[Y_{i2}^u\mid D_i=0] - \mathbb{E}[Y_{i1}^\ell\mid D_i=0]\right)
        \bigg].
    \end{align*}
Furthermore, the sharp identified set of the ATT $\theta_{\mathtt{ATT}}$ is given by
    \begin{align*}
        \Theta_{\mathtt{SPT}} = \aE{\bm{Y}_{i2} \mid D_i=1} \ominus \mathcal{M}_{\mathtt{SPT}}.
    \end{align*}
\end{proposition}

When the scalar outcome is observable (i.e., $Y^\ell_{it}=Y^u_{it}$), $\Theta_{\mathtt{SPT}}$ reduces to the conventional DID formula. Although this identification strategy is fully compatible with the standard assumption, the resulting identified set is generally not very informative, as it tends to be excessively wide. The reason for these wide bounds is that the lower bound of $\mathcal{M}_{\mathtt{SPT}}$ is constructed by considering an extremely unfavorable scenario in which the lower bound for the treated group (NJ) follows a trend connecting the upper bound for the control group (PA) in period $t=1$ to the lower bound for PA in period $t=2$. The upper bound of $\mathcal{M}_{\mathtt{SPT}}$ tends to be large for a similar reason.

Figure \ref{fig: naive bounds} illustrates this point using the \cite{Card_Krueger:1994} study. We can see that $\mathcal{M}_{\mathtt{SPT}}$ is extremely wide. Moreover, this approach fails to retain important features of the control group, such as the fact that the identified set shrinks over time and that both the upper and lower bounds shift downward.
\begin{figure}[t]
    \centering
    \begin{subfigure}[b]{0.475\textwidth}
        \centering
        \includegraphics[width=\linewidth]{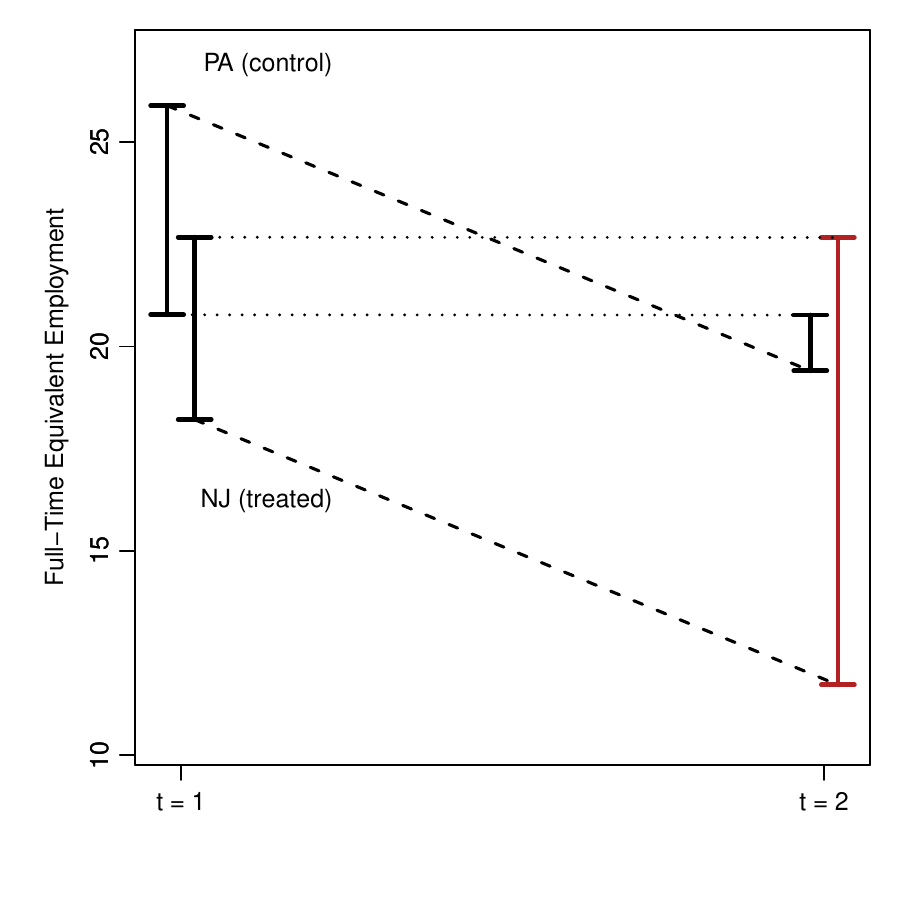}
        \caption{$\mathcal{M}_{\mathtt{SPT}}$}
        \label{fig: naive bounds}
    \end{subfigure}
    \hfill
    \begin{subfigure}[b]{0.475\textwidth}
        \centering
        \includegraphics[width=\linewidth]{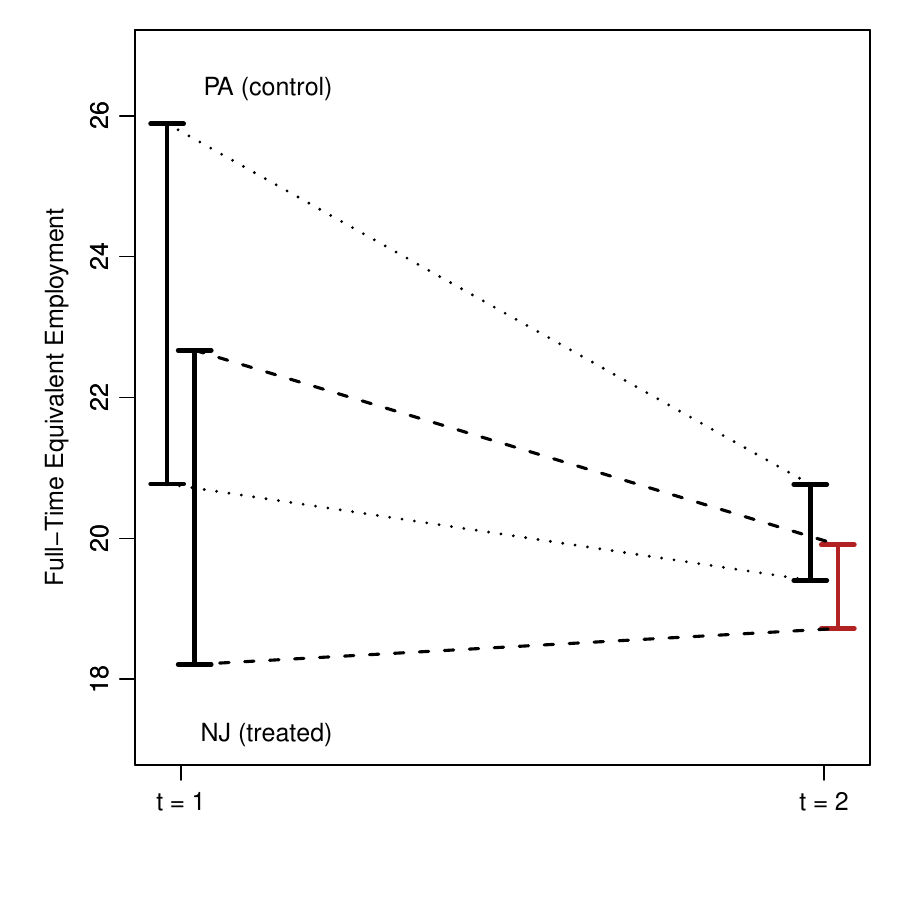}
        \caption{$\mathcal{M}_{\mathtt{IPT}}$}
        \label{fig: PT bounds}
    \end{subfigure}
    \caption{Empirical bounds using the \cite{Card_Krueger:1994} data}
    \label{fig: bounds}
\end{figure}

\subsection{Strategy 2: Interval parallel trends}
The previous analysis motivates us to consider another notion of parallel trends for interval-valued outcomes. To this aim, we begin with the scalar parallel trends assumption (Assumption Scalar-PT). This assumption can be interpreted as stating that there exists a mapping that transports $\E{Y_{i1}(0)\mid D_i=0}$ to $\E{Y_{i2}(0)\mid D_i=0}$ and that the same mapping also transports $\E{Y_{i1}(0)\mid D_i=1}$ to the counterfactual mean $\E{Y_{i2}(0)\mid D_i=1}$.

This idea can be extended to the interval-data case. 
Consider the mapping that takes the lower bound $\E{Y_{i1}^\ell(0)\mid D_i=0}$ to $\E{Y_{i2}^\ell(0)\mid D_i=0}$ and the upper bound $\E{Y_{i1}^u(0)\mid D_i=0}$ to $\E{Y_{i2}^u(0)\mid D_i=0}$. There is a unique linear map $T(\cdot)$ achieving this mapping, which is given by
\begin{align}
    T(x) =
    \frac{\mathbb{E}[Y_{i2}^u\mid D_i=0] - \mathbb{E}[Y_{i2}^\ell\mid D_i=0]}{\mathbb{E}[Y_{i1}^u\mid D_i=0] - \mathbb{E}[Y_{i1}^\ell\mid D_i=0]}\left(x-\mathbb{E}[Y_{i1}^\ell\mid D_i=0]\right) + \mathbb{E}[Y_{i2}^\ell\mid D_i=0].\label{eq: PT map}
\end{align}
Let $T([a,b])=[T(a),T(b)]$ be the mapping of an interval $[a,b]$ by $T$. Then we can see that
\begin{align*}
    \aE{\bm{Y}_{i2}(0) \mid D_i = 0} = T\left(\aE{\bm{Y}_{i1}(0) \mid D_i = 0}\right),
\end{align*}
i.e., $T$ transports the potential Aumann mean $\aE{\bm{Y}_{i1}(0) \mid D_i = 0}$ at $t=1$ to its counterpart at $t=2$, $\aE{\bm{Y}_{i2}(0) \mid D_i = 0}$. Therefore, by applying the same transport for the conditional Aumann means given $D_i=1$, the parallel trends assumption could be extended to the interval-valued outcome case as follows.
\begin{asm-IPT}
    Let $T$ be the map defined in \eqref{eq: PT map}.
    It holds that
    \begin{align*}
        \aE{\bm{Y}_{i2}(0) \mid D_i = 1} = 
        T\left(\aE{\bm{Y}_{i1}(0) \mid D_i = 1}\right).
    \end{align*}
\end{asm-IPT}
Under this assumption, the sharp identified set of the counterfactual mean and ATT are characterized as follows.
\begin{proposition}[Sharp identified set under interval parallel trends]\label{prop: PT bounds}
Consider the setup of this section. Suppose Assumption Interval-PT holds true and $\P{Y_{i1}^u> Y_{i1}^\ell\mid D_i=0}>0$. Then the sharp identified set of $\E{Y_{i2}(0) \mid D_i=1}$ is given by
    \begin{align*}
        \mathcal{M}_{\mathtt{IPT}}&=
        T\left(\aE{\bm{Y}_{i1} \mid D_i = 1}\right).
    \end{align*}
Furthermore, the sharp identified set of the ATT $\theta_{\mathtt{ATT}}$ is given by
    \begin{align*}
        \Theta_{\mathtt{IPT}} = \aE{\bm{Y}_{i2} \mid D_i = 1} \ominus \mathcal{M}_{\mathtt{IPT}}.
    \end{align*}
\end{proposition}

Assumption Interval-PT seems a natural extension of the standard parallel trends assumption to the interval data case. However, this identification strategy does not work well in practice. Figure \ref{fig: PT bounds} shows the identified set $\mathcal{M}_{\mathtt{IPT}}$ under Assumption Interval-PT. Although the issue of $\mathcal{M}_{\mathtt{SPT}}$ under Assumption Scalar-PT being excessively wide is resolved, $\mathcal{M}_{\mathtt{IPT}}$ yields an increasing trend in the lower bound for the treated group (NJ), even though the lower bound for the control group (PA) exhibits a decreasing trend.

To resolve this issue, one might be tempted to apply the standard parallel trends assumption separately to the lower and upper bounds, that is, to project the control group's lower (resp. upper) bound trend from $t=1$ to $t=2$ onto the treated group's lower (resp. upper) bound. However, this strategy only works when the pre-treatment intervals are of comparable length. When the lengths of the bounds differ in period $t=1$, treating the upper and lower bounds as independent objects and relying solely on their sums or differences is unappealing, as this discards the information contained in the difference in their length. Moreover, and more importantly from an empirical perspective, this approach does not guarantee that the resulting image remains a valid interval.\footnote{For example, suppose that $\aE{\bm{Y}_{i1} \mid D_i=0} = [0,3]$, $\aE{\bm{Y}_{i2} \mid D_i=0} = [2,3]$, and $\aE{\bm{Y}_{i1} \mid D_i=1} = [0,1]$. For the control group, the lower bound increases by 2, while the upper bound remains unchanged. If we apply these trends to $[0,1]$, the resulting ``interval'' becomes $[2,1]$, which is not a valid interval.}

In sum, researchers need to be cautious of a naive application or extension of the parallel trends assumption to interval outcomes, which may yield uninformative or counterintuitive results. The next section presents an alternative identification strategy, which exhibits desirable properties for the DID analysis with interval outcomes.

\section{Parallel shifts}\label{sec:shift}
Motivated by the previous discussion, we now develop an alternative notion of the parallel trends assumption for interval-valued outcomes. 
We first recall that the interval parallel trends assumption (Assumption Interval-PT) is motivated by extending the standard interpretation of the conventional scalar parallel trends (Assumption Scalar-PT), i.e., the expected change in the bounds of $\E{Y_{it}(0)\mid D_i=1}$ over time is identical to the one of $\E{Y_{it}(0)\mid D_i=0}$. 
However, it should be noted that Assumption Scalar-PT admits another, equally natural representation:
\begin{align*}
    \E{Y_{i2}(0) \mid D_i=1} - \E{Y_{i2}(0) \mid D_i=0}
    =
    \E{Y_{i1}(0) \mid D_i=1} - \E{Y_{i1}(0) \mid D_i=0},
\end{align*}
i.e., the shift in the conditional mean of $Y_{it}(0)$ from the control to the treatment group is identical over time. 
This representation is mathematically equivalent to Assumption Scalar-PT, but motivates an alternative notion to the parallel trends assumption in the interval-valued outcome setting.

In particular, as a counterpart of the right-hand side of the above representation, we consider the linear mapping $S(\cdot)$ such that\begin{align*}
    \aE{\bm{Y}_{i1}(0) \mid D_i=1}= 
    S\left(\aE{\bm{Y}_{i1}(0) \mid D_i=0}\right),
\end{align*}
which is uniquely defined as
\begin{align}
    S(y) = \frac{\mathbb{E}[Y_{i1}^u\mid D_i=1] - \mathbb{E}[Y_{i1}^\ell\mid D_i=1]}{\mathbb{E}[Y_{i1}^u\mid D_i=0] - \mathbb{E}[Y_{i1}^\ell\mid D_i=0]}\left(y-\mathbb{E}[Y_{i1}^\ell\mid D_i=0]\right)+\mathbb{E}[Y_{i1}^\ell\mid D_i=1].\label{eq:S}
\end{align}
We can then formulate an alternative notion of the parallel trends assumption as follows.
\begin{asm-PS}[Parallel shifts]\label{assumption: interval PS}
    Let $S$ be the map defined in \eqref{eq:S}. It holds that
    \begin{align*}
        \aE{\bm{Y}_{i2}(0)\mid D_i=1} = 
        S\left(\aE{\bm{Y}_{i2}(0)\mid D_i=0}\right).
    \end{align*}
\end{asm-PS}

Assumption PS says that the map $S$ transporting the bounds for the control group in period $t=1$ to those for the treatment group in period $t=1$ also transports the former in period $t=2$ to the latter in period $t=2$.\footnote{When the scalar outcome is observable (i.e., $Y^\ell_{it}=Y^u_{it}$), this assumption reduces to the conventional scalar parallel trends in Assumption Scalar-PT by regarding $(\mathbb{E}[Y_{i1}^u\mid D_i=1] - \mathbb{E}[Y_{i1}^\ell\mid D_i=1])/(\mathbb{E}[Y_{i1}^u\mid D_i=0] - \mathbb{E}[Y_{i1}^\ell\mid D_i=0]) = 1$, i.e., by treating the ``lengths'' of the sigletons as being equal. More generally, when the lengths in the pretreatment period are identical, the map $S$ reduces to the bound-by-bound approach discussed in the previous footnote.} Based on this assumption, our main result is presented as follows.
\begin{proposition}[Sharp identified set under parallel shifts]\label{prop: PS bounds}
    Consider the setup outlined in Section \ref{sec:setup}. Suppose Assumption PS holds true and $\P{Y_{i1}^u> Y_{i1}^\ell\mid D_i=0}>0$. Then the sharp identified set of the counterfactual mean $\E{Y_{i2}(0) \mid D_i=1}$ is given by
    \begin{align*}
        \mathcal{M}_{\mathtt{PS}} = S\left(\aE{\bm{Y}_{i2} \mid D_i=0}\right).
    \end{align*}
    Furthermore, the sharp identified set of the ATT $\theta_{\mathtt{ATT}}$ is given by
    \begin{align*}
       \Theta_{\mathtt{PS}} = \aE{\bm{Y}_{i2} \mid D_i=1}\ominus \mathcal{M}_{\mathtt{PS}} \eqqcolon [\vartheta_L, \vartheta_U],
    \end{align*}
    where
    \begin{align*}
        \vartheta_L &\coloneqq \E{Y_{i2}^\ell\mid D_i=1} - \frac{\E{Y_{i1}^u-Y_{i1}^\ell\mid D_i=1}}{\E{Y_{i1}^u-Y_{i1}^\ell\mid D_i=0}}\E{Y_{i2}^u-Y_{i1}^\ell\mid D_i=0}-\E{Y_{i1}^\ell\mid D_i=1},\\
        \vartheta_U &\coloneqq \E{Y_{i2}^u\mid D_i=1} - \frac{\E{Y_{i1}^u-Y_{i1}^\ell\mid D_i=1}}{\E{Y_{i1}^u-Y_{i1}^\ell\mid D_i=0}}\E{Y_{i2}^\ell-Y_{i1}^\ell\mid D_i=0}-\E{Y_{i1}^\ell\mid D_i=1}.
    \end{align*}
\end{proposition}

We treat the outcome as an interval-valued variable, which does not reside in a vector space. Therefore, extrapolation under the interval parallel trends (Assumption Interval PT) is generally different from that under the parallel shifts assumption.

\begin{figure}[ht]
  \centering

  \begin{subfigure}[b]{0.7\textwidth}
    \centering
    \begin{minipage}[b]{0.45\textwidth}
      \centering
      \includegraphics[width=\linewidth]{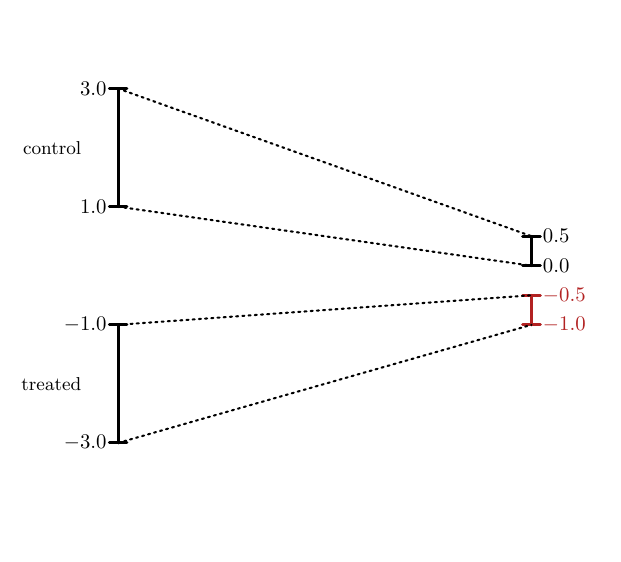}
    \end{minipage}\hfill
    \begin{minipage}[b]{0.45\textwidth}
      \centering
      \includegraphics[width=\linewidth]{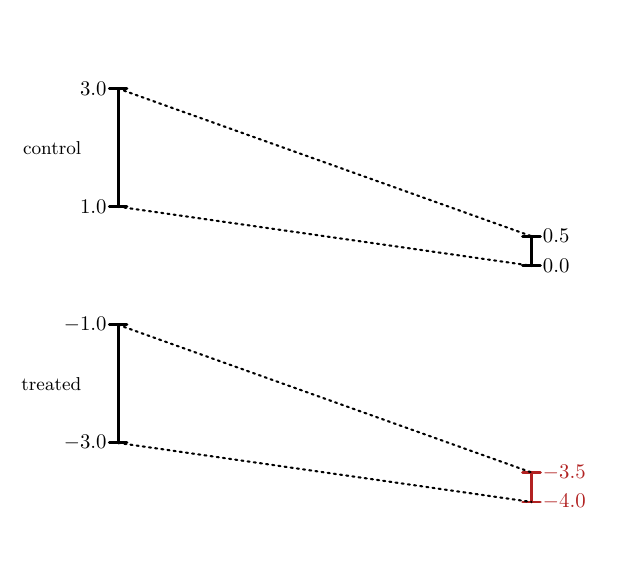}
    \end{minipage}
    \caption{Example 1}
    \label{fig:ex1}
  \end{subfigure}
  \begin{subfigure}[b]{0.7\textwidth}
    \centering
    \begin{minipage}[b]{0.45\textwidth}
      \centering
      \includegraphics[width=\linewidth]{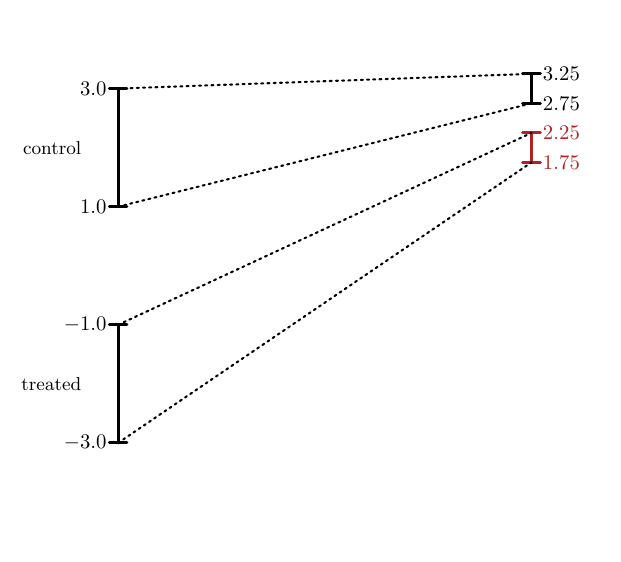}
    \end{minipage}\hfill
    \begin{minipage}[b]{0.45\textwidth}
      \centering
      \includegraphics[width=\linewidth]{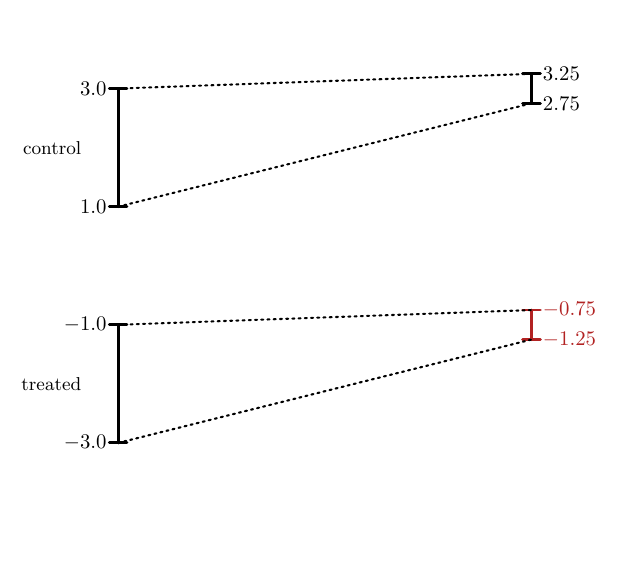}
    \end{minipage}
    \caption{Example 2}
    \label{fig:ex2}
  \end{subfigure}
  \begin{subfigure}[b]{0.7\textwidth}
    \centering
    \begin{minipage}[b]{0.45\textwidth}
      \centering
      \includegraphics[width=\linewidth]{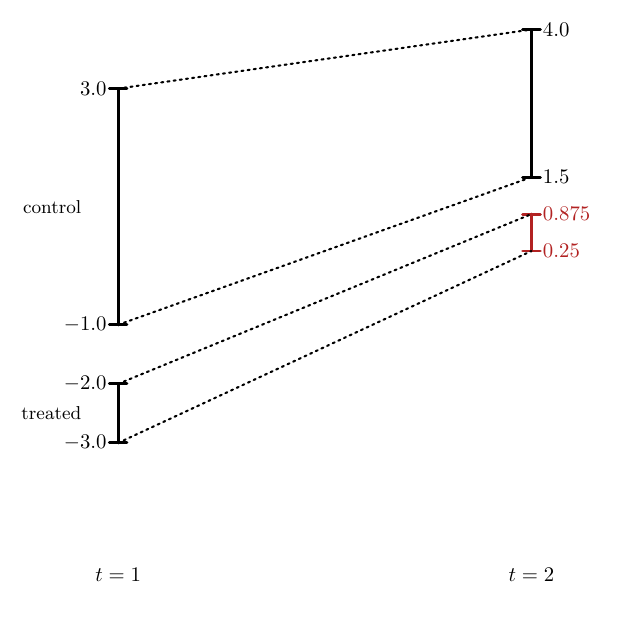}
    \end{minipage}\hfill
    \begin{minipage}[b]{0.45\textwidth}
      \centering
      \includegraphics[width=\linewidth]{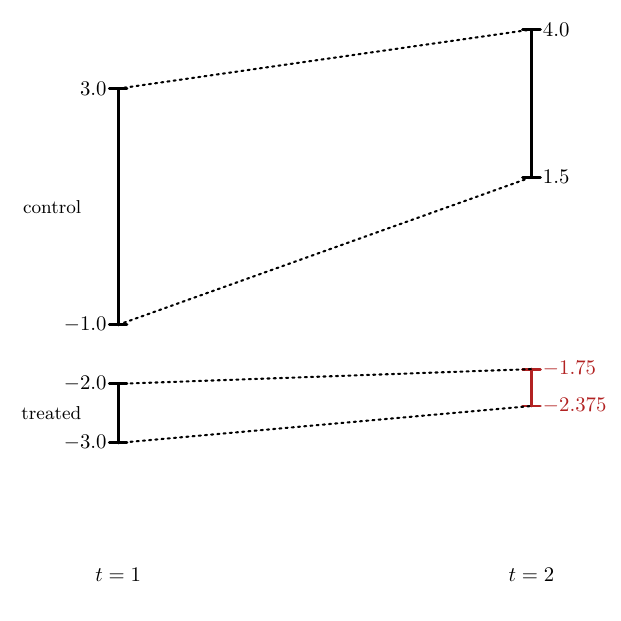}
    \end{minipage}
    \caption{Example 3}
    \label{fig:ex3}
  \end{subfigure}

  \caption{Numerical Examples}
  \label{fig:example}

  \begin{flushleft}
    \footnotesize
    \renewcommand{\baselineskip}{11pt}
    \textbf{Note:} In the left panel, $\mathcal{M}_{\texttt{IPT}}$ under Assumption Interval-PT is shown, 
    while $\mathcal{M}_{\texttt{PS}}$ under Assumption PS is shown in the right panel.
  \end{flushleft}
\end{figure}
To illustrate this point, Figure \ref{fig:example} provides numerical examples. In the left panel, $\mathcal{M}_{\texttt{IPT}}$ under Assumption Interval-PT is shown, while $\mathcal{M}_{\texttt{PS}}$ under Assumption PS is shown in the right panel. In Example 1 (Figure \ref{fig:ex1}), we consider $\aE{\bm{Y}_{i1}\mid D_i=0} = [1.0, 3.0]$, $\aE{\bm{Y}_{i1}\mid D_i=1} = [-3.0, -1.0]$, and $\aE{\bm{Y}_{i2}\mid D_i=0} = [0.0, 0.5]$. Similarly to Figure \ref{fig: PT bounds}, Assumption Interval-PT produces a counterintuitive prediction, showing an upward trend that runs counter to the downward trend of the control group. In contrast, Assumption PS shows a more sensible behavior that respects the downward trend of the control group. In Example 2 (Figure \ref{fig:ex2}), we use the same pre-treatment intervals but modify $\aE{\bm{Y}_{i2}\mid D_i=0}$ to $[2.75, 3.25]$. Again, Assumption Interval-PT produces seemingly erratic behavior: while the upper bound for the control group moves only marginally, that for the treated group shifts more dynamically. This unpleasant feature is once again resolved by Assumption PS. A similar pattern emerges when the lengths of the bounds differ in the pre-treatment period. In Example 3 (Figure \ref{fig:ex3}), we consider the case where $\aE{\bm{Y}_{i1}\mid D_i=0} = [-1.0, 3.0]$, $\aE{\bm{Y}_{i1}\mid D_i=1} = [-3.0, -2.0]$, and $\aE{\bm{Y}_{i2}\mid D_i=0} = [1.5, 4.0]$. Once again, under Assumption Interval-PT, the upper bound moves too dynamically relative to that of the control group, whereas under Assumption PS, the relatively marginal shifts in the upper bound compared to the lower bound are preserved.

\subsection{Property of parallel shifts}
The previous numerical examples suggest the practical advantages of the parallel shifts approach. 
To further motivate this approach, we consider some desirable properties that the identification assumption for the DID with interval outcomes should satisfy. For intervals $\bm{A}_j=[a_j^\ell,a_j^u]$ and $\bm{B}_j=[b_j^\ell,b_j^u]$ ($j=1,2$), consider a map $M:(\bm{A}_1,\bm{A}_2,\bm{B}_1)\mapsto \bm{B}_2$. Let $|\bm{A}|$ be the length of an interval $\bm{A}$. We consider the following conditions as desirable properties of the map $M$.
\begin{cond}[Desirable property of $M$] \quad
\begin{description}
    \item[(i)] $b_2^\ell \leq b_2^u$.
    \item[(ii)] ${|\bm{B}_2|}/{|\bm{B}_1|} = {|\bm{A}_2|}/{|\bm{A}_1|}$.
    \item[(iii)] $b_2^u - b_1^u = \gamma (a_2^u - a_1^u)$ and $b_2^\ell - b_1^\ell = \gamma (a_2^\ell - a_1^\ell)$ for some $\gamma>0$.
\end{description}
\end{cond}

Condition (i) is a weak requirement so that the mapped value $\bm{B}_2$ is a proper interval. Condition (ii) is motivated to exclude bounds like $\mathcal{M}_{\mathtt{SPT}}$ under Assumption Scalar-PT as illustrated in Figure \ref{fig: naive bounds}. It requires that the changes in the size of the identified set evolve in a similar manner across the control and treatment groups, ensuring that both groups display consistent patterns in how their identified sets expand or shrink over time. Condition (iii) is motivated to exclude bounds such as $\mathcal{M}_{\mathtt{IPT}}$ under Assumption Interval-PT. This condition imposes two restrictions. First, $\gamma>0$ requires that the lower and upper bounds move in the same direction for both the control and treatment groups, which excludes the possibility of Figures \ref{fig: PT bounds} and \ref{fig:ex1}. In addition, it requires that movements in the upper and lower bounds be equally respected. For example, it rules out situations in which, relative to the control group, the lower bound for the treatment group moves only half as much, while the upper bound moves twice as much. This excludes the possibility of Figures \ref{fig:ex2} and \ref{fig:ex3}.\footnote{Condition (iii) can also be motivated from the perspective of the random set theory. As discussed in \citet[p.~76]{MoMo18}, every closed interval $\bm{A}$ is uniquely defined by its support function $s_{\bm{A}}(\cdot)$. Then we can see that Condition (iii) is equivalent to $s_{\bm{B}_2}(u) - s_{\bm{B}_1}(u) = \gamma \left\{s_{\bm{A}_2}(u) - s_{\bm{A}_1}(u)\right\}$. In this sense, Condition (iii) can be understood as a parallel transport in a space of the support functions.}

Indeed, in the class of continuous maps $M$ (in the sense that continuous at the boundaries of each element), only the parallel shift map based on $S$ in (\ref{eq:S}) can meet Conditions (i)-(iii).
\begin{lemma}[Desirability of parallel shifts] Suppose that $|\bm{A}_1|,|\bm{B}_1|>0$. A continuous map $M$ satisfies Conditions (i)-(iii) if and only if $M(\bm{A}_1,\bm{A}_2,\bm{B}_1)=[S(a_2^\ell), S(a_2^u)]$, where
\begin{align*}
S(y) = \frac{|\bm{B}_1|}{|\bm{A}_1|}(y - a_1^\ell) + b_1^\ell.
\end{align*}
\end{lemma}

This lemma provides a theoretical justification for Assumption PS and clarifies the reason behind the favorable behavior observed in Figure \ref{fig:example}.

\subsection{Covariates}
It is straightforward to extend our identification strategy using the parallel shifts to accommodate covariates $X_i\in\mathbb{R}^d$. 
By modifying the parallel trends assumption, the sharp identified sets are characterized in the same manner. 
\begin{asm-CPS}[Conditional parallel shifts]\label{assumption: conditional interval PS}
    It holds that
    \begin{align*}
        \aE{\bm{Y}_{i2}(0)\mid D_i=1, X_i} = 
        S\left(\aE{\bm{Y}_{i2}(0)\mid D_i=0, X_i}; X_i\right),
    \end{align*}
    where
    \begin{align*}
        &S(y; X_i)\\
        &\quad= 
        \frac{\mathbb{E}[Y_{i1}^u\mid D_i=1, X_i] - \mathbb{E}[Y_{i1}^\ell\mid D_i=1, X_i]}{\mathbb{E}[Y_{i1}^u\mid D_i=0, X_i] - \mathbb{E}[Y_{i1}^\ell\mid D_i=0, X_i]}\left(y-\mathbb{E}[Y_{i1}^\ell\mid D_i=0, X_i]\right)+\mathbb{E}[Y_{i1}^\ell\mid D_i=1, X_i].
    \end{align*}
\end{asm-CPS}
\begin{proposition}\label{prop: conditional PS bounds}
    Consider the setup outlined in Section \ref{sec:setup}. Suppose Assumption Conditional-PS holds tule, $\P{Y_{i1}^u> Y_{i1}^\ell\mid D_i=0, X_i}>0$, and $\P{D_i=1 \mid X_i}\in (0,1)$. Then the sharp identified set of the counterfactual mean $\E{Y_{i2}(0) \mid D_i=1}$ is given by
    \begin{align*}
        \mathcal{M}_{\mathtt{CPS}}=
        \aE{S\left(\aE{\bm{Y}_{i2} \mid D_i=0, X_i}; X_i\right) \mid D_i=1}.
    \end{align*}
    Furthermore, the sharp identified set of the ATT $\theta_{\mathtt{ATT}}$ is given by
    \begin{align*}
    \Theta_{\mathtt{CPS}} = \aE{\bm{Y}_{i2} \mid D_i=1}\ominus \mathcal{M}_{\mathtt{CPS}}.
    \end{align*}
\end{proposition}

\subsection{Implementation}
Our proposed bounds $\Theta_{\mathtt{PS}}$ are straightforward to implement. Recalling the explicit expression of $\vartheta_L$ and $\vartheta_U$ provided in Proposition \ref{prop: PS bounds}, every conditional expectation can be consistently estimated by its sample analog, which yields a sample analog estimator $\hat{\vartheta}_L$ and $\hat{\vartheta}_U$. Under random sampling and finite second moments of $Y_{it}^\ell$ and $Y_{it}^u$, $\hat{\vartheta}_L$ and $\hat{\vartheta}_U$ are jointly asymptotically normally distributed.

Furthermore, by construction, it holds that $\hat{\vartheta}_U \geq \hat{\vartheta}_L$ with probability one. Hence, the superconsistency condition of \citet[Assumption 3]{Stoye:2009} is satisfied in our case, allowing us to conduct statistical inference based on the confidence interval of \cite{Imbens_Manski:2004}. When discrete covariates are available, the same procedure applies. When the covariates are continuously distributed, some form of (semi-)parametric modeling becomes necessary; see \cite{Chang:2020} and \cite{SZ20} for related discussion. As the endpoints of $\Theta_{\mathtt{CPS}}$ have simple forms, we conjecture that a similar procedure is applicable in our setting as well.

As in the scalar-outcome case, plausibility of the parallel shifts assumption can be guided by goodness-of-fit in pretreatment periods, provided that multiple pretreatment periods are available. In the interval-valued setting, one can examine if the predicted lower and upper bounds track their realized counterparts.\footnote{Although this approach is widely used, formally testing such ``pre-trends'' can distort the main test as emphasized by \cite{Roth:2022}. As an alternative, bounded-variation-type assumptions introduced by \cite{Manski_Pepper:2018} and extended by \cite{Rambachan_Roth:2023} can also be incorporated into the interval-outcome framework. We view such extensions as promising avenues for future work on DID with interval-valued outcomes.}

\section{Card and Krueger (1994)}\label{sec:CK}
\begin{figure}[t]
    \begin{center}
    \begin{subfigure}[b]{0.475\textwidth}
        \centering
        \includegraphics[width=\linewidth]{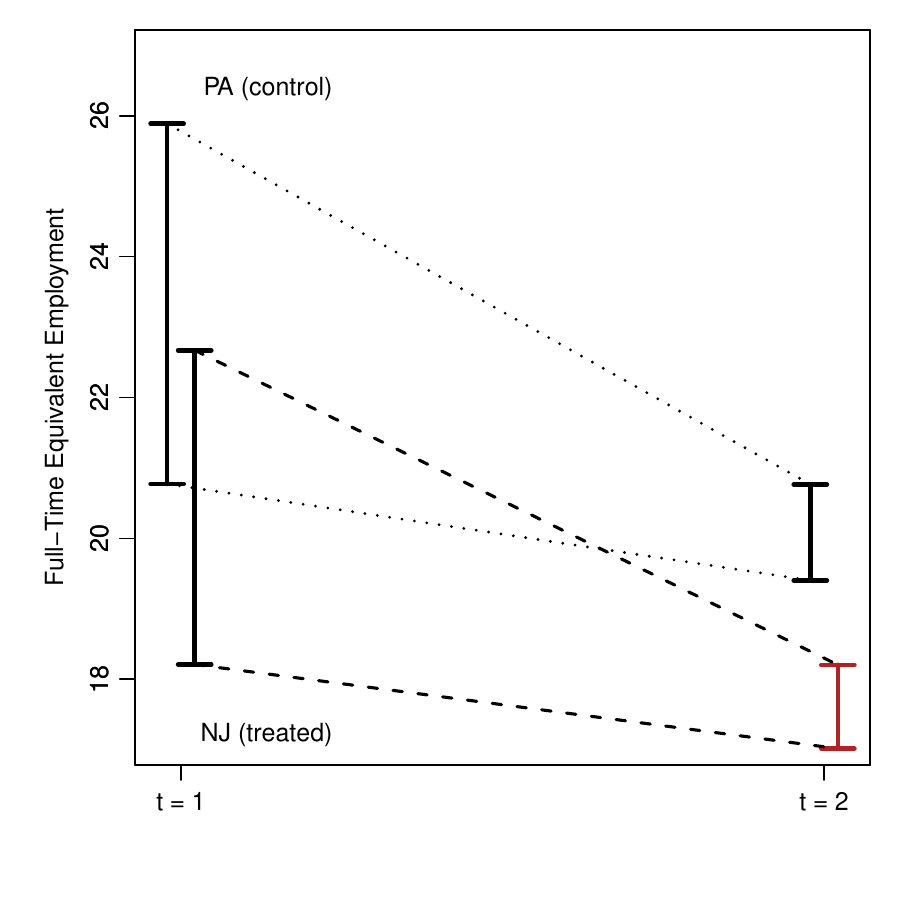}
        \caption{$\mathcal{M}_{\mathtt{PS}}$}
        \label{fig: PS bounds}
    \end{subfigure}
    \hfill
    \begin{subfigure}[b]{0.475\textwidth}
        \centering
        \includegraphics[width=\linewidth]{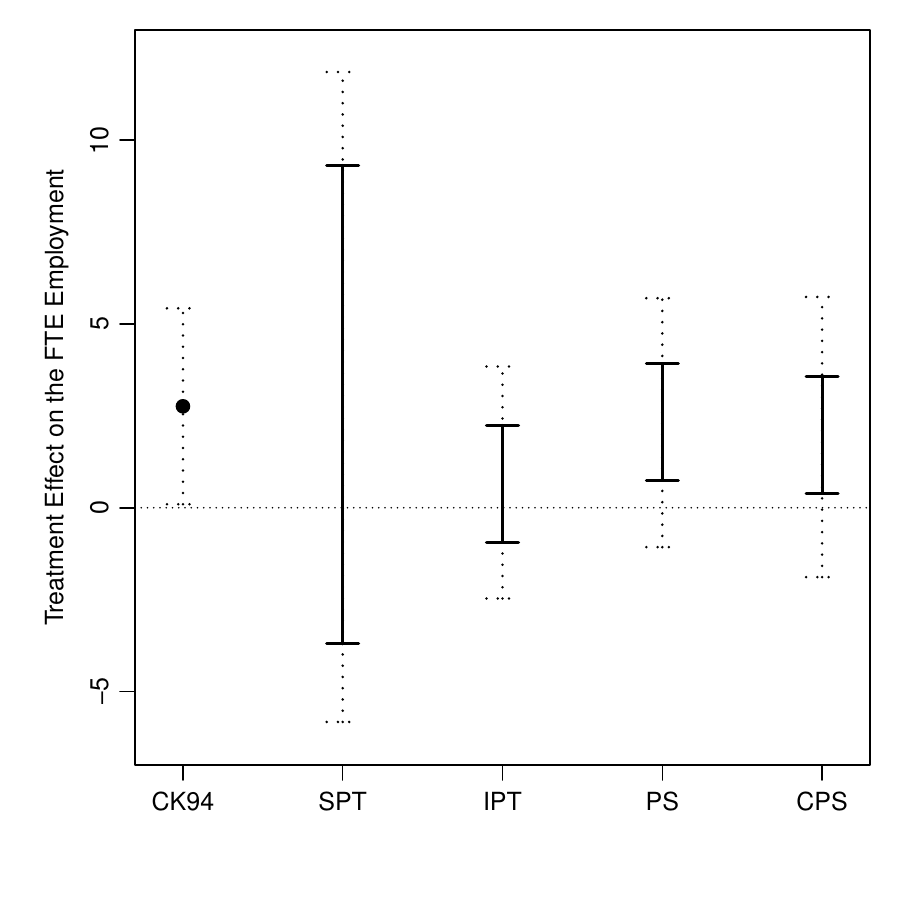}
        \caption{Bounds on $\theta_{\mathtt{ATT}}$}
        \label{fig: bounds on att}
    \end{subfigure}
    \caption{Empirical results from the reanalysis of \cite{Card_Krueger:1994}}
    \label{fig: empirical results}
    \end{center}

    {\footnotesize
    \noindent\begin{minipage}{\textwidth}
    \renewcommand{\baselineskip}{11pt}
    \textbf{Note:} In Panel (B), the solid lines represent the bounds, while the dotted lines depict the 95\% confidence intervals. For the bounds, we employ the \cite{Imbens_Manski:2004} and \cite{Stoye:2009} confidence intervals, using plug-in analytical variance estimators.
    In CPS, we follow \cite{Card_Krueger:1994} and use the restaurant chain and an indicator for whether the store is company-owned as covariates.
    \end{minipage}
    }
\end{figure}

We are now in a position to reexamine the \cite{Card_Krueger:1994} minimum wage study. We begin by estimating the counterfactual interval $\mathcal{M}_{\mathtt{PS}}$ identified under Assumption PS. Figure \ref{fig: PS bounds} displays the estimated intervals. We can see that the identified set is tighter than $\mathcal{M}_{\mathtt{SPT}}$ in Figure \ref{fig: naive bounds}. Furthermore, the trends in the boundaries of the identified set for NJ exhibit a similar pattern to that of PA. In particular, the lower bound shows a clear decreasing trend, in contrast to the behavior of $\mathcal{M}_{\mathtt{IPT}}$ in Figure \ref{fig: PT bounds}. These findings illustrate the empirical attractiveness of our proposed notion of parallel shifts.

Figure \ref{fig: bounds on att} summarizes the bounds on the ATT under the various assumptions. CK94 reproduces the estimate of \cite{Card_Krueger:1994}. SPT presents the identified set under Assumption Scalar-PT, and we again confirm that this set is too wide to be informative. The bounds under Assumption Interval-PT, reported as IPT, address this issue, although the identified set still includes zero. In contrast, the identified sets under Assumptions PS and CPS both suggest positive treatment effects, although in every case the estimates are statistically indistinguishable from zero.

The results based on Assumptions PS and CPS are well aligned with the later study of \cite{Card_Krueger:2000}, which reanalyzes the original 1994 study using an administrative data free from survey or measurement error. They find smaller positive point estimates, though these are not statistically significant. 
The bounds and associated confidence intervals are consistent with this finding. In this sense, our results reaffirm the original findings reported in their studies (\citealp{Card_Krueger:1994, Card_Krueger:2000}) and further highlight the potential usefulness of our proposed procedure.

\section{Conclusion}
This paper has extended the scope of the difference-in-differences methodology toward the case of interval outcomes. Our analyses on sharp identified sets for the average treatment effect on the treated using the scalar- and interval-parallel trends assumptions revealed some cautions for practitioners, i.e., the resulting identified sets can be uninformative or counterintuitive. We then propose an alternative parallel shifts assumption to overcome these drawbacks. An empirical example based on the \cite{Card_Krueger:1994} minimum wage study convincingly endorses its practical value. 
Although our focus has been on the canonical DID setting, it is of interest to study analogous issues of interval outcomes in more general settings, such as staggered DID \citep{CS21}, fuzzy DID \citep{deChaisemartin:2018}, and bounded-variation type assumptions \citep{Manski_Pepper:2018, Rambachan_Roth:2023}. Extending our framework to these cases would further enhance its practical applicability.

\newpage
\appendix
\section{Proofs}
\noindent Let $A\oplus B = \{a+b: a\in A\,\text{ and }\, b\in B\}$ be the Minkowski sum of intervals $A$ and $B$.
\begin{proof}[Proof of Proposition \ref{prop: naive bounds}]
    We begin with the proof of the validity and sharpness of $\mathcal{M}_{\mathtt{SPT}}$.
    By \eqref{eq:Y} and Assumption Scalar-PT,
    \begin{align*}
        \E{Y_{i2}(0) \mid D_i=1} &= 
        \E{Y_{i1}(0)\mid D_i=1} + \E{Y_{i2}(0) - Y_{i1}(0)\mid D_i=0}\\
        &=
        \E{Y_{i1}\mid D_i=1} + \left(\E{Y_{i2}\mid D_i=0} - \E{Y_{i1}\mid D_i=0}\right).
    \end{align*}
    Recalling that $Y_{i1}\in\bm{Y}_{i1}$, we have $\E{Y_{i1} \mid D_i=1}\in \aE{\bm{Y}_{i1}\mid D_i=1}$. We have also seen that $\E{Y_{it} \mid D_i=0}\in\aE{\bm{Y}_{it}\mid D_i=0}$. Therefore, we have
    \begin{align*}
        \E{Y_{i2}(0) \mid D_i=1} &\in
        \aE{\bm{Y}_{i1}\mid D_i=1} \oplus 
        \left(\aE{\bm{Y}_{i2}\mid D_i=0}\ominus\aE{\bm{Y}_{i1}\mid D_i=0}\right)\\
        &=
        \bigg[
        \mathbb{E}[Y_{i1}^\ell\mid D_i=1] + \left(\mathbb{E}[Y_{i2}^\ell\mid D_i=0] - \mathbb{E}[Y_{i1}^u\mid D_i=0]\right),\\
        &\qquad\quad 
        \mathbb{E}[Y_{i1}^u\mid D_i=1] + \left(\mathbb{E}[Y_{i2}^u\mid D_i=0] - \mathbb{E}[Y_{i1}^\ell\mid D_i=0]\right)
        \bigg] = \mathcal{M}_{\mathtt{SPT}}.
    \end{align*}
    To see that the lower bound is attainable, define the following data-generating process (DGP) of potential outcomes. 
    For units with $D_i=0$, set $Y_{i1}(0) = Y_{i1}^u$ and $Y_{i2}(0) = Y_{i2}^\ell$ so that $Y_{it}(0)\in\bm{Y}_{it}$ and $\mathbb{E}[Y_{i2}(0)-Y_{i1}(0)\mid D_i=0] = \mathbb{E}[Y_{i2}^\ell - Y_{i1}^u \mid D_i=0]$.
    For units with $D_i=1$, set $Y_{i1}(0) = Y_{i1}^\ell$ and $Y_{i2}(0) = Y_{i1}^\ell + \mathbb{E}[Y_{i2}^\ell - Y_{i1}^u \mid D_i=0]$.
    Then Assumption Scalar-PT holds. 
    Furthermore, under this construction, the lower bound of $\mathcal{M}_{\mathtt{SPT}}$ is indeed attained.    
    A symmetric argument shows that the upper bound is attainable. The interior points are attained by considering the convex mixtures of these two DGPs. 
    Hence, $ = \mathcal{M}_{\mathtt{SPT}}$ is the sharp identified set of $\E{Y_{i2}(0) \mid D_i=1}$.
    
    Furthermore, in the main text, we have seen that $\E{Y_{i2} \mid D_i=1}\in\aE{\bm{Y}_{i2} \mid D_i=1}$, implying the validity of the bounds $\aE{\bm{Y}_{i2} \mid D_i=1} \ominus \mathcal{M}_{\mathtt{SPT}} = \Theta_{\mathtt{SPT}}$.
    Choose any DGP of $(Y_{i1}(0), Y_{i2}(0), D_i)$ attaining a point $\theta_1$ in $\mathcal{M}_{\mathtt{SPT}}$ (such that it is observationally equivalent to the data).
    Likewise, choose any DGP of $(Y_{i2}(1), D_i)$ attaining a point $\theta_2$ in $\aE{\bm{Y}_{i2} \mid D_i=1}$.
    Then, we can combine these two marginal DGPs into a joint DGP in which $(Y_{i1}(0), Y_{i2}(0))$ and $Y_{i2}(1)$ are independent conditional on $D_i=1$. Under this DGP, the point $\theta_2 - \theta_1 (\in \aE{\bm{Y}_{i2} \mid D_i=1}\ominus \mathcal{M}_{\mathtt{SPT}})$ is attained.
    As $\theta_1$ and $\theta_2$ are arbitrary, the sharpness is proved.
\end{proof}

\begin{proof}[Proof of Proposition \ref{prop: PT bounds}]
    The map $T$ in \eqref{eq: PT map} is well-defined by $\P{Y_{i1}^u > Y_{i1}^\ell \mid D_i=0}>0$. By \eqref{eq:Y} and Assumption Interval-PT, the validity of the bounds follows immediately.
    To see the attainability of the lower bound, define the following DGP. For $D_i = 0$, set arbitrarily $Y_{it}(0) \in \bm{Y}_{it}$. 
    For $D_i=1$, set $Y_{i1}(0) = Y_{i1}^\ell \in \bm{Y}_{i1}$, $Y_{i2}(0) = T(Y_{i1}^\ell)$, and $\bm{Y}_{i2}(0) = T(\bm{Y}_{i1}(0))$.
    The Aumann mean and the map $T$ are both linear in the boundary points of the interval, we have that $\aE{\bm{Y}_{i2}(0) \mid D_i=1} = \aE{T(\bm{Y}_{i1}(0)) \mid D_i=1} = T(\aE{\bm{Y}_{i1}(0) \mid D_i=1})$, that is, Assumption Interval-PT is satisfied.
    Moreover, using linearity again, we have that $\E{Y_{i2}(0) \mid D_i=1} = \mathbb{E}[T(Y_{i1}^\ell) \mid D_i=1] = T(\mathbb{E}[Y_{i1}^\ell \mid D_i=1])$, which equals the lower bound of $\mathcal{M}_{\mathtt{IPT}}$.
    The sharpness of $\mathcal{M}_{\mathtt{IPT}}$ then follows by a symmetric argument, as in the proof of the previous proposition.
    Likewise, the sharpness of $\Theta_{\mathtt{IPT}}$ follows.
\end{proof}

\begin{proof}[Proof of Proposition \ref{prop: PS bounds}]
    The proof is similar to that of Proposition \ref{prop: PT bounds}.
\end{proof}

\begin{proof}[Proof of Proposition \ref{prop: conditional PS bounds}]
    The statement follows from Proposition \ref{prop: PS bounds} and the law of iterated expectation.
\end{proof}

\begin{proof}[Proof of Lemma]
    ($\Leftarrow$): Note that $a_2^\ell \leq a_2^u$ by definition.
    Then we have that $b_2^\ell = S(a_2^\ell) \leq S(a_2^u) = b_2^u$, showing that Condition (i) is satisfied.
    Next, we can compute that $|\bm{B}_2| = S(a_2^u) - S(a_2^\ell) = |\bm{B}_1|\cdot|\bm{A}_2|/|\bm{A}_1|$, which implies Condition (ii).
    Finally, note that
    \begin{align*}
        b_2^u - b_1^u &= S(a_2^u) - b_1^u 
        = \frac{|\bm{B}_1|}{|\bm{A}_1|} \left(a_2^u - a_1^\ell\right) + b_1^\ell - b_1^u\\
        &= \frac{|\bm{B}_1|}{|\bm{A}_1|} \left(a_2^u - a_1^u + |\bm{A}_1|\right) - |\bm{B}_1| = \frac{|\bm{B}_1|}{|\bm{A}_1|} \left(a_2^u - a_1^u\right).
    \end{align*}
    A similar manipulation shows that $b_2^\ell - b_1^\ell = ({|\bm{B}_1|}/{|\bm{A}_1|}) (a_2^\ell - a_1^\ell)$. Hence, Condition (iii) is satisfied with $\gamma = {|\bm{B}_1|}/{|\bm{A}_1|}$.
    
    \noindent ($\Rightarrow$): We begin with the case where $|\bm{A}_1|\neq |\bm{A}_2|$. Conditions (i) and (iii) imply that $|\bm{B}_2|-|\bm{B}_1| = \gamma (|\bm{A}_2|-|\bm{A}_1|)$, and thus
    \begin{align*}
        \frac{|\bm{B}_2|}{|\bm{B}_1|}-1= \gamma \frac{|\bm{A}_2| - |\bm{A}_1|}{|\bm{B}_1|}.
    \end{align*}
    Then, by Condition (ii), we have that
    \begin{align*}
        \frac{|\bm{A}_2|-|\bm{A}_1|}{|\bm{A}_1|} = \gamma \frac{|\bm{A}_2| - |\bm{A}_1|}{|\bm{B}_1|}.
    \end{align*}
    With the assumption $|\bm{A}_1|\neq |\bm{A}_2|$, we obtain that $\gamma = |\bm{B}_1|/|\bm{A}_1|\eqqcolon \gamma_1$, so that
    \begin{align*}
        b_2^\ell = \gamma_1 (a_2^\ell - a_1^\ell) + b_1^\ell,
    \end{align*}
    and
    \begin{align*}
        b_2^u = \gamma_1 (a_2^u - a_1^u) + b_1^u
        = \gamma_1(a_2^u - a_1^\ell) + b_1^\ell + \underbrace{|\bm{B}_1| - \gamma_1 |\bm{A}_1|}_{=0} 
        = \gamma_1(a_2^u - a_1^\ell) + b_1^\ell.
    \end{align*}
    Hence, the map is given by $S(x) = |\bm{B}_1|/|\bm{A}_1|(x - a_1^\ell) + b_1^\ell$.
    When $|\bm{A}_1|=|\bm{A}_2|$, let $\bm{A}_2^{(m)}\coloneqq [a_2^\ell-1/m, a_2^u+1/m]$ and note that $\bm{B}_2^{(m)} \coloneqq M(\bm{A}_1,\bm{A}_2^{(m)},\bm{B}_1)$ satisfies $\bm{B}_2^{(m)} = S(\bm{A}_2^{(m)})$. Then we can see that $\bm{B}_2 = \lim_{n\to\infty}\bm{B}_2^{(m)} = S(\lim_{n\to\infty} \bm{A}_2^{(m)})=S(\bm{A}_2)$, where the first equality uses the continuity of $M$, and the lemma is proved.
\end{proof}

\bibliographystyle{apalike} 
\bibliography{refs}

\end{document}